\documentclass{fundam}

\publyear{22}
\papernumber{2102}
\volume{185}
\issue{1}

\usepackage{algpseudocode}

\usepackage{url} 
\usepackage{graphicx}
\usepackage{tikz}
\usetikzlibrary{automata, positioning, arrows}

\begin{document}

	\title{Analyzing Network Robustness via Residual Closeness}
	
	\address{Dokuz Eylul University, hande.tuncel@deu.edu.tr}
	
	\author{Hande Tuncel Golpek\\
		Maritime Faculty \\
		Dokuz Eylul University, Izmir, Turkiye\\
		hande.tuncel@deu.edu.tr
		\and Mehmet Ali Bilici\\
		Computer Engineering Department\\
		Ege University,
		Izmir, Turkiye\\ mehmet.ali.bilici@ege.edu.tr \and Aysun Aytac\\
		Department of Mathematics\\
		Ege University, 
		Izmir, Turkiye\\ aysun.aytac@ege.edu.tr} 
	
	\maketitle
	
	\runninghead{H. Tuncel Golpek, M.A. Bilici, A. Aytac}{Analyzing Network Robustness...}
	
	\begin{abstract}
		Networks are inherently vulnerable to vertex failures, making the analysis of their structural robustness a fundamental problem in graph theory. In this study, we investigate the closeness and vertex residual closeness of graphs, with a particular focus on the middle graph representations of certain special graph classes, which provide a richer structural framework for analysis.
		We derive exact expressions for the closeness values of these middle graphs and determine their residual closeness under vertex failures. By utilizing results obtained from specific graph families, we establish several general bounds for broader graph classes. Furthermore, by exploiting the relationship between the closeness of a graph, its line graphs and middle graphs, we obtain new results that relate these three structures. In addition, we present a computational procedure for evaluating closeness and vertex residual closeness in middle graphs, together with numerical verification of selected theoretical results and an analysis of its computational complexity.
			\end{abstract}
	
	\begin{keywords}
		network vulnerability; network design and communication; residual closeness
	\end{keywords}

\section*{Introduction}
Closeness centrality is a fundamental concept in network analysis, measuring how efficiently a vertex can communicate with all other vertices in a graph. It has been widely used in various applications such as social networks, communication systems, and transportation networks. The classical definition of closeness centrality of a vertex is based on the inverse of the sum of its distances to all other vertices in the graph.

Classical graph-theoretic measures—such as connectivity, toughness, and integrity \cite{Chvatal,Barefoot,Harary2}—have traditionally been used to evaluate network vulnerability. However, these metrics often fall short when it comes to capturing partial degradations, particularly in cases where a network remains technically connected but suffers from impaired performance.

In this context, closeness centrality has emerged as an important alternative, as it quantifies how efficiently a vertex can interact with the rest of the network through shortest-path distances. Nevertheless, the classical definition of closeness is inherently limited to connected graphs and may fail to provide meaningful information in the presence of disconnected components \cite{Freeman}. Recognizing this limitation, Latora and Marchiori introduced the concept of efficiency \cite{Latora}, broadening the applicability of closeness to disconnected or degraded networks.

Subsequently, Dangalchev proposed a modified version of closeness that resolves interpretation issues in disconnected graphs and provides a smoother degradation response to failures \cite{Dangalchev}. Building on this formulation, he further introduced residual closeness, a more sensitive and practical metric that quantifies how the failure of nodes or links affects the efficiency of communication paths across the network—even when the network remains connected.

From a design and optimization perspective, residual closeness enables the identification of structurally critical components—nodes or edges—whose failure would lead to significant declines in overall communication efficiency, even if the network remains formally connected. This insight is especially valuable in real-world applications where resilient, cost-effective, and high-performance architectures must balance robustness with resource constraints.

The flow and clarity of the paper will be supported by providing certain graph-theoretic terms and notations. We follow two primary references \cite{Chartrand,Harary}. Let $G$ be a simple, finite, and undirected graph, with vertex set $V(G)$ and edge set $E(G)$.
In this paper the vertex set of a graph is labeled as $\{v_{1},...v_{n}\}$
with n vertices. The open neighborhood of a vertex $v_{i}$, denoted by $%
N(v_{i})$, is defined as $N(v_{i})=\{v_{j}\in V(G):(v_{i}v_{j})\in E(G)\},$ $%
i,j\in \{1,...,n\}$. The degree of a vertex $v_{i}$, written as $\deg
(v_{i}) $, is the number of vertices adjacent to $v_{i}$. The distance
between two vertices $v_{i}$ and $v_{j}$, denoted by $d(v_{i},v_{j})$, is
the length of the shortest path connecting them, which is also called the $%
v_{i}-v_{j}$ geodesic. Additionally, three key graph parameters related to
distance are eccentricity, diameter, and radius. In a connected graph, the
eccentricity of a vertex $v_{i}$, denoted by $ec(v_{i})$, is the
largest distance from $v_{i}$ to any other vertex. The diameter of a graph $%
G,$ expressed as $diam(G)=$ $\max \{ec(v_{i}):i\in \{1,...,n\}\}$, is
the greatest eccentricity among all vertices. Similarly, the radius of $G$,
denoted by $rad(G)$, is the smallest eccentricity among the vertices of the
graph.

In this paper, we are interested in the useful closeness formula for vertex $%
v_{j}$ $(1\leq j\leq n)$ and invariant of the parameter. We begin with the parameter closeness, denoted by $C(G)$. For a vertex
$v_j \in V(G)$, its closeness is defined as
\[
C(v_j)=\sum_{v_i \neq v_j}\frac{1}{2^{d(v_i,v_j)}},
\]
where $d(v_i,v_j)$ denotes the shortest distance between the vertices
$v_i$ and $v_j$. The closeness of the graph $G$ is defined as
\[
C(G)=\sum_{v_j \in V(G)} C(v_j).
\]
Therefore,
\[
C(G)=\sum_{v_j \in V(G)}
\sum_{v_i \neq v_j}
\frac{1}{2^{d(v_i,v_j)}}.
\] Additionally, based on the closeness
parameter, Dangalchev introduced more refined parameter that is
\textit{vertex residual closeness} \cite{Dangalchev}. The main aim of vertex residual closeness is to evaluate how removing a vertex affects the network's vulnerability. The closeness value after removing vertex $v_{k}$, represented as $C_{k}$, is computed using the formula 
\[
C_{k}=\sum\limits_{v_{j}}\sum\limits_{v_{i}\neq v_{j}}\frac{1}{%
	2^{d_{v_{k}}(v_{i},v_{j})}}, 
\]%
where $d_{v_{k}}(v_{i},v_{j})$ is the distance between vertices $v_{i}$ and $%
v_{j}$ after vertex $v_{k}$ has been removed from the graph. Furthermore,
the vertex residual closeness, denoted by $R$, is given as 
$R=\min_{k}\{C_{k}\}.$ For more detailed information, the reader is referred to the following references \cite{AytacOdabas1,AytacOdabas2,AytacOdabas3,AytacOdabas4,AytacTuraci1,DangGeneralized, DangSplitting, DangThorn, Dangalchev, Dangline, OdabasAytac,TuraciOkten, TuraciVAytac}.

In this work, we will concentrate on a significant graph structure known as the middle graph. 
\[
V(M(G))=V(G)\cup \{v_e:e\in E(G)\},
\]
where $v_e$ is the vertex corresponding to the edge $e$ of $G$.
A vertex $u\in V(G)$ is adjacent to $v_e$ if and only if $u$ is
incident with $e$ in $G$, and two edge-vertices $v_e$ and $v_f$
are adjacent if and only if the corresponding edges $e$ and $f$
are adjacent in $G$.
This structure is introduced in \cite{HamadaYoshimura} as a type of intersection graph. In addition the line graph concept is used in the paper as a tool. Line graph is obtained as adding a vertex in line graph, $L(G),$ for each edge in $G$ and this new vertices in $L(G)$ are adjacent if their orginated edges  have common vertices in $G.$

By analyzing residual closeness in middle graph structures derived from various base graphs, this paper offers new characterizations of network vulnerability and robustness. These findings contribute not only to the theoretical understanding of residual closeness, but also to its practical applications in network analysis, graph-based optimization, and the development of fault-tolerant systems in diverse domains.

In the first section, the paper will present known results and information derived from the literature used in the study. The second section will cover the computation of closeness values for the considered structures. In the third section, we will provide the results for the vertex residual closeness values  In the last section, we will provide some general results about parameters that is considered. Moreover, we develop a computational procedure for determining the closeness and residual closeness values of middle graphs and examine its computational complexity. This algorithmic perspective complements our theoretical results and facilitates the practical evaluation of closeness-based measures. Overall, our results enhance the theoretical understanding of closeness and residual closeness under vertex failures and provide a framework for analyzing the robustness of complex graph structures.
\section{Known Results}
In this section, to enhance the clarity and comprehensibility of the paper,
we will provide some essential information and certain results available in
the literature that are necessary for our calculations. First, the formula
for the geometric series and its derivative, which are frequently used in
our calculations, will be provided. Afterward, we will recall some
fundamental results related to closeness, link residual closeness, and
vertex residual closeness that we will use in our paper.

The geometric sum is:

\begin{equation}\label{eq1}
	\sum\limits_{i=0}^{n-1}r^{i}=\frac{r^{n}-1}{r-1}
\end{equation}

and its derivative is:

\begin{equation}\label{eq2}
\frac{d}{dr}(\sum\limits_{i=0}^{n-1}r^{i})=\sum\limits_{i=1}^{n-1}ir^{i-1}=\frac{nr^{n-1}}{r-1}-\frac{r^{n}-1}{%
	(r-1)^{2}}. 
\end{equation}

Now, we can present some fundamental and useful results related to our paper:

\begin{theorem}\label{zeynep}
	\cite{Dangalchev,AytacOdabas1,OdabasAytac}The closeness of 
	\begin{itemize}	
	\item[i)] the complete graph with $n$ vertices is $C(K_{n})=\frac{n(n-1)}{2},$
	
	\item[ii)] the star graph with $n+1$ vertices is $C(S_{1,n})=\frac{n(n+3)}{4}$,
	
	\item[iii)] the path graph with $n$ vertices is        
    $C(P_{n})=2n-4+2^{2-n}$,
	
	\item[vi)]the cycle graph with $n$ vertices is\\           $C(C_{n})=\left\{ 
	\begin{array}{cc}
		2n(1-1/2^{(n-1)/2}) & ,\text{if }n\text{ is odd} \\ 
		n(2-3/2^{n/2}) & ,\text{if }n\text{ is even}%
	\end{array}%
	\right. .$
    \end{itemize}
	
\end{theorem}

\begin{theorem}
	\cite{Dangalchev} The vertex residual closeness of the complete graph $K_{n}$ is \\$%
	R(K_{n})=((n-1)(n-2))/2.$
\end{theorem}

\begin{theorem}\label{line}
	\cite{Dangline} The closeness of some special line graphs are
	
	(i) $C(L(C_{n}))=C(C_{n}),$
	
	(ii) $C(L(P_{n}))=C(P_{n-1})=2n-6+2^{3-n},$
	
	(iii) $C(L(S_{1,n}))=C(K_{n})=\frac{n(n-1)}{2},$
	
	(iv) $C(L(K_{n}))=\frac{n(n^{3}+2n^{2}-13n+10)}{16}.$
\end{theorem}

\section{Closeness of Some Middle Graphs}
In this section, we investigate some results about the closeness of some special types of middle graphs. Throughout this section, the vertices inherited from $G$ are denoted
by $u$-vertices and the vertices corresponding to the edges of $G$
are denoted by $v$-vertices. We use $C(u\sim u)$ for the total
closeness contribution of ordered pairs of distinct $u$-vertices,
$C(v\sim v)$ for that of ordered pairs of distinct $v$-vertices,
and $C(u\sim v)$ for the contribution of ordered pairs whose first
vertex is a $u$-vertex and whose second vertex is a $v$-vertex.

\begin{theorem}\label{C1}
	Let $G$ be a path graph with $n\geq2$ vertices, $P_{n}$. The closeness value of the	middle path graph is
\[
C(M(P_{n}))=7n-16+\frac{18}{2^{n}}. 
\]
\end{theorem}
\begin{proof}
Let the vertex set of the middle graph $M(P_n)$ be
$
V(M(P_n))=
\{u_1,u_2,\ldots,u_n,v_1,v_2,\ldots,v_{n-1}\},
$
where $u_i$, $1\leq i\leq n$, are the vertices of $P_n$, and
$v_j$, $1\leq j\leq n-1$, is the new vertex inserted into the edge
$u_ju_{j+1}$.

Since the closeness of a graph is obtained by summing the contributions
over all ordered pairs of distinct vertices, we divide the computation
into three types of vertex pairs: pairs of $u$-vertices, pairs of
$v$-vertices, and pairs consisting of one $u$-vertex and one
$v$-vertex.

Using splitting the graph as in \cite{DangSplitting},
the closeness value can be calculated more easily. First, consider two vertices $u_i$ and $u_j$ with $i<j$.  For \(i<j\), let \(k=j-i\). Then
\(d_{M(P_n)}(u_i,u_j)=k+1\), and for a fixed \(k\), there are
\(n-k\) pairs satisfying \(i<j\). Since the closeness is computed over ordered pairs, each unordered pair
contributes twice. Let us denote this total value as $C(u\sim u)$.	
\begin{align*}
	C(u\sim u)&=2(\frac{n-1}{2^{2}}+\frac{n-2}{2^{3}}+\cdots +\frac{1}{2^{n}}%
	)\\&=2\sum\limits_{i=1}^{n-1}\frac{n-i}{2^{i+1}}. 
	\end{align*}
	
	Next, the vertices $v_1,v_2,\ldots,v_{n-1}$ induce the path
	$P_{n-1}$. The total closeness value is the closeness value of the 	$P_{n-1}$ graph, provided in Theorem \ref{zeynep} (iii). Let us denote this total value as $C(v\sim v)$.%
	\[
	C(v\sim v)=C(P_{n-1})=2(n-1)-4+\frac{1}{2^{(n-3)}}. 
	\]%
Finally, consider pairs consisting of one vertex $u_i$ and one vertex
$v_j$. Let us denote this total value as $C(u\sim v)$. For each $i$, there are $2(n-i)$ pairs
$(u_j,v_k)$ whose distance is $i$.
\begin{align*}
C(u\sim v)&=2(\frac{n-1}{2}+\frac{n-2}{2^{2}}+\cdots +\frac{1}{%
	2^{n-1}})\\
    &=2\sum\limits_{i=1}^{n-1}\frac{n-i}{2^{i}}. 
\end{align*}

The pairs $(v_j,u_i)$ give the same contribution. Thus, the total
contribution of the mixed pairs is $2C(u\sim v)$.

Consequently,
\begin{eqnarray*}
	C(M(P_{n}))
	&=&C(u\sim u)+C(v\sim v)+2C(u\sim v)\\ &=&5\sum\limits_{i=1}^{n-1}\frac{n-i}{2^{i}}+2(n-1)-4+\frac{1}{2^{n-3}}.
\end{eqnarray*}

Using the geometric sum formula in equation \eqref{eq1} and its
derivative in equation \eqref{eq2}, we obtain

\begin{eqnarray*}
	C(M(P_{n})) &=&5(n-2)+\frac{5}{2^{n-1}}+2(n-1)-4+\frac{1}{2^{n-3}} \\
	&=&7n-16+\frac{18}{2^{n}}.
\end{eqnarray*}
\end{proof}

\begin{theorem}\label{M(Cn)}
	Let G be a cycle graph with $n\geq3$ vertices, $C_{n}$. The closeness value of
	the middle cycle graph is
    {\small
\[
C(M(C_{n}))= 
\begin{cases}
	\frac{3}{2}C(C_{n})+4n(1-(\frac{1}{2})^{\frac{n}{2}}), & \text{if } n\text{
		is even } \\ 
	\frac{3}{2}C(C_{n})+2n(2(1-(\frac{1}{2})^{\frac{n-1}{2}})+(\frac{1}{2})^{%
		\frac{n+1}{2}}), & \text{if }n\text{ is odd}%
\end{cases}%
\]%
}
where $C(C_{n})=\left\{ 
\begin{array}{cc}
	2n(1-1/2^{(n-1)/2}) & ,\text{if }n\text{ is odd} \\ 
	n(2-3/2^{n/2}) & ,\text{if }n\text{ is even}%
\end{array}%
\right. .$
\end{theorem}
\begin{proof}
	Let's label the set of vertices of the middle cycle graph as $\{u_{1},u_{2},\dots,u_{n},v_{1}, v_{2},\dots ,v_{n}\}$. Here, $u_{i}$, for 
	$1\leq i\leq n$, represents the vertices of the $C_{n}$ graph, and $v_{j}$,
	for $1\leq j\leq n$, are vertices added later according to the definition of
	the middle graph. To obtain the closeness value, we can relate the vertices
	of the graph as follows:
	
	\begin{enumerate}
		\item The relationship among vertices $u_{i}$, $1\leq i\leq n$, total
		closeness value represented as $C(u\sim u)$.
		
		\item The relationship among vertices $v_{j}$, $1\leq j\leq n$, total
		closeness value represented as $C(v\sim v)$.
		
		\item The relationship between vertices $u_{i}$ and $v_{j}$ and total
		closeness value represented as $C(u\sim v)$.
	\end{enumerate}
	
	Thus, we can split the graph and calculate the closeness value. Due to the
	structure of the middle cycle graph, it is calculated as
	$
	C(u\sim u)=\frac{1}{2}C(C_n).
	$
	On the other hand, the $v$-vertices induce a cycle isomorphic to
	$C_n$, and therefore
	$
	C(v\sim v)=C(C_n).
	$
	Consequently, total closeness value of first two cases is obtained as $3/2$ of the closeness value of a cycle
	graph with $n$ vertices, denoted as $C(C_{n})$. The closeness value of the
	cycle graph is $C(C_{n})$ is known from Theorem \ref{zeynep}.
	
	The third case involves calculating the closeness value related to the
	relationship between the vertices of the cycle graph and the vertices added
	later to create the middle graph. Let us denote this value as $C(u\sim v)$. Here, the distance between vertices $u$ and $v$ will vary depending on
	the diameter of the cycle graph. Therefore, considering the distances
	between vertices and depending on whether $n$ is even or odd, we will
	examine two cases:
	
	Case1: If $n$ is odd;%
	\begin{eqnarray*}
		C(u\sim v) &=&n(\frac{2}{2^{1}}+\frac{2}{2^{2}}+\cdots +\frac{2}{%
			2^{(n-1)/2}}+(\frac{1}{2})^{(n+1)/2}) \\
		&=&n(2\sum\limits_{i=1}^{\frac{n-1}{2}}\frac{1}{2^{i}}+(\frac{1}{2}%
		)^{(n+1)/2}).
	\end{eqnarray*}%
	In this case, the value of $C(u\sim v)$ can be calculated using the
	geometric series formula in equation \eqref{eq1}:%
	\[
	2C(u\sim v)=2n(2(1-(\frac{1}{2})^{(n-1)/2})+(\frac{1}{2})^{(n+1)/2}). 
	\]%
	
	Case2: If $n$ is even;%
	\[
	C(u\sim v)=n(\frac{2}{2^{1}}+\frac{2}{2^{2}}+\cdots +\frac{2}{2^{(%
			\frac{n}{2})}})=n(2\sum\limits_{i=1}^{\frac{n}{2}}\frac{1}{2^{i}}). 
	\]%
	In this case, the value of $C(u\sim v)$ also can be calculated using
	the geometric series formula in equation \eqref{eq1}:%
	\[
	C(u\sim v)=2n(1-(\frac{1}{2})^{(n/2)}) 
	\]%
	Due to Case 1 and Case 2, the total closeness value for the middle cycle
	graph, denoted as $C(M(C_{n})$, is calculated using the following formula:%
	\[
	C(M(C_{n}))=C(u\sim u)+C(v\sim v)+2C(u\sim v). 
	\]%
	Therefore, we get the desired result
	where $C(C_{n})$ is provided in Theorem \ref{zeynep}.
\end{proof}

\begin{theorem}\label{M(S)}
	Let $G$ be a star graph with $n+1$ vertices, $S_{1,n}$ where $n\geq2$. The closeness value
	of the middle star graph is
\[
C(M(S_{1,n}))=\frac{(9n^{2}+11n)}{8}. 
\]
\end{theorem}
\begin{proof}
	Let the vertex set of the middle graph $M(S_{1,n})$ be $V(M(S_{1,n}))=\{u_1,u_2,\ldots,u_{n+1},v_1,v_2,\ldots,v_n\}$, where $u_i$, $1\leq i\leq n+1$, are the vertices of $S_{1,n}$, with
$u_1$ being the central vertex, and $v_j$, $1\leq j\leq n$, are the
new vertices inserted into the edges of $S_{1,n}$.

Since the closeness of a graph is obtained by summing the contributions
over all ordered pairs of distinct vertices, we divide the computation
into three types of vertex pairs: pairs of $u$-vertices, pairs of
$v$-vertices, and pairs consisting of one $u$-vertex and one
$v$-vertex.

First, consider the pairs of $u$-vertices. The central vertex $u_1$
has distance $2$ from each of the $n$ leaf vertices. Moreover, each
$n$ leaf vertex has distance $2$ from the central vertex and distance $3$
from the other $(n-1)$ leaf vertices. Thus, the total contribution of
the $u$-vertices is
\[
C(u\sim u)
=\frac{n}{2}+\frac{n(n-1)}{2^3}.
\]

Next, the vertices $v_1,v_2,\ldots,v_n$ induce the complete graph $K_n$. The total closeness value of  \(K_{n}\), provided in Theorem \ref{zeynep} (i),
and therefore
\[
C(v\sim v)=C(K_n)=\frac{n(n-1)}{2}.
\]

Finally, consider pairs consisting of one $u$-vertex and one
$v$-vertex. The central vertex $u_1$ is at distance $1$ from all
$v$-vertices. Each leaf vertex is at distance $1$ from its
corresponding $v$-vertex and at distance $2$ from the remaining
$n-1$ vertices of the set $\{v_1,\ldots,v_n\}$. Thus,
\[
C(u\sim v)
=n+\frac{n(n-1)}{2^2}.
\]
The pairs $(v_j,u_i)$ give the same contribution. Hence, the total
contribution of the mixed pairs is $2C(u\sim v)$. Consequently,
\begin{align*}
	C(M(S_{1,n}))
	&=C(u\sim u)+C(v\sim v)+2C(u\sim v)\\
	&=\frac{n}{2}+\frac{n(n-1)}{2^3}
	+\frac{n(n-1)}{2}
	+2\left(n+\frac{n(n-1)}{2^2}\right)\\
	&=\frac{9n^2+11n}{8}.
\end{align*}
\end{proof}

\begin{theorem}
	Let $G$ be a complete graph with $n$ vertices, $K_{n},\  n\geq3$. The closeness value
	of the  middle complete graph is
\[
C(M(K_{n}))=\frac{n(n-1)(n^{2}+7n+2)}{16}. 
\]
\end{theorem}
\begin{proof}
	Let the vertex set of the middle graph $M(K_n)$ be $V(M(K_n))=
	\{u_1,u_2,\ldots,u_n,v_1,v_2,\ldots,v_{\frac{n(n-1)}{2}}\},
	$ where $u_i$, $1\leq i\leq n$, are the vertices of $K_n$, and
	$v_j$, $1\leq j\leq \frac{n(n-1)}{2}$, are the new vertices inserted
	into the edges of $K_n$.
	
	We can split the graph to easily calculate the closeness value. Since the closeness of a graph is obtained by summing the contributions
	over all ordered pairs of distinct vertices, we divide the computation
	into three types of vertex pairs: pairs of $u$-vertices, pairs of
	$v$-vertices, and pairs consisting of one $u$-vertex and one
	$v$-vertex.
	
	First, consider the pairs of $u$-vertices. According to definition of middle graph, distance of 
	$u$-vertices in $M(K_n)$ is
	$2$. Since there are $n(n-1)$ ordered pairs of distinct
	$u$-vertices, we obtain
	\[
	C(u\sim u)=\frac{n(n-1)}{2^2}.
	\]
	
	Next, consider the vertices $v_1,v_2,\ldots,
	v_{\frac{n(n-1)}{2}}$. Each $v$-vertex is adjacent to $2(n-2)$ other
	$v$-vertices, while all remaining $v$-vertices are at distance $2$.
	Therefore,
	\[
	C(v\sim v)
	=
	\frac{n(n-1)}{2}
	\left(
	\frac{2(n-2)}{2}
	+
	\frac{\frac{n(n-1)}{2}-1-2(n-2)}{2^2}
	\right).
	\]
	Hence,
	\[
	C(v\sim v)
	=
	\frac{n(n-1)\bigl(4(n-1)+n(n-1)-6\bigr)}{2^4}.
	\]
	
	Finally, consider pairs consisting of one $u$-vertex and one
	$v$-vertex. Each $u$-vertex is adjacent to $n-1$ vertices of the
	$v$-set, while the remaining
	$\frac{n(n-1)}{2}-(n-1)$ vertices are at distance $2$. Thus,
	\[
	C(u\sim v)
	=
	n\left(
	\frac{n-1}{2}
	+
	\frac{\frac{n(n-1)}{2}-(n-1)}{2^2}
	\right)
	=
	\frac{n(n-1)(n+2)}{2^3}.
	\]
	The pairs $(v_j,u_i)$ give the same contribution. Hence, the total
	contribution of the mixed pairs is $2C(u\sim v)$. Consequently,
	\begin{align*}
		C(M(K_n))
		&=C(u\sim u)+C(v\sim v)+2C(u\sim v)\\
		&=\frac{n(n-1)}{2^2}
		+\frac{n(n-1)\bigl(4(n-1)+n(n-1)-6\bigr)}{2^4}\\
		&\qquad
		+2\left(\frac{n(n-1)(n+2)}{2^3}\right)\\
		&=\frac{n(n-1)(n^2+7n+2)}{16}.
	\end{align*}
\end{proof}
\begin{theorem} \label{MW1n}
	Let $G$ be a wheel graph with $n+1$ vertices, $W_{1,n}$ provided that $n>4$.
	The closeness value of the middle wheel graph is%
	\[
	C(M(W_{1,n}))=2n^{2}+6n. 
	\]
\end{theorem}

\begin{proof}
Let the vertex set of the middle graph $M(W_{1,n})$ be
$
V(M(W_{1,n}))=
\{u_1,u_2,\ldots,u_{n+1},v_1,v_2,\ldots,v_{2n}\},
$
where $u_1$ is the central vertex of $W_{1,n}$ and
$u_2,\ldots,u_{n+1}$ are the vertices of the cycle. The vertices
$v_1,\ldots,v_n$ correspond to the edges incident with the central
vertex, while $v_{n+1},\ldots,v_{2n}$ correspond to the edges of the
cycle.

Since the closeness of a graph is obtained by summing the contributions
over all ordered pairs of distinct vertices, we divide the computation
into three types of vertex pairs: pairs of $u$-vertices, pairs of
$v$-vertices, and pairs consisting of one $u$-vertex and one
$v$-vertex.

First, consider the pairs of $u$-vertices. The central vertex $u_1$
has distance $2$ from all other $n$ vertices. Each non-central
vertex has distance $2$ from its three adjacent vertices in
$W_{1,n}$ and distance $3$ from the remaining $n-3$ vertices.
Therefore,
\[
C(u\sim u)=\frac{n^2+5n}{8}.
\]

Next, consider pairs of $v$-vertices. The first $n$ vertices
$v_1,\ldots,v_n$ induce $K_n$, while the last $n$ vertices
$v_{n+1},\ldots,v_{2n}$ induce $C_n$. The total contribution from
the vertices of $K_n$ is $\frac{3n^2}{4}$, and the contribution from
the vertices of $C_n$ is $\frac{n(3n+11)}{8}.$
Consequently,
\[
C(v\sim v)
=
\frac{3n^2}{4}+\frac{n(3n+11)}{8}.
\]

Finally, consider the pairs consisting of one $u$-vertex and one
$v$-vertex. The central vertex contributes $\frac{3n}{4}$, while
each non-central $u$-vertex has distance $1$ from three
$v$-vertices, distance $2$ from $n+1$ vertices, and distance $3$
from the remaining $n-4$ vertices. Thus,
\begin{align*}
	C(u\sim v)
	&=\frac{3n}{4}
	+(n-1)\left(
	\frac{3}{2}+\frac{n+1}{4}+\frac{n-4}{8}
	\right)\\
	&=\frac{n(3n+16)}{8}.
\end{align*}
The pairs $(v_j,u_i)$ give the same contribution. Hence, the total
contribution of the mixed pairs is $2C(u\sim v)$. Therefore,
\begin{align*}
	C(M(W_{1,n}))
	&=C(u\sim u)+C(v\sim v)+2C(u\sim v)\\
	&=\frac{n^2+5n}{8}
	+\frac{3n^2}{4}
	+\frac{n(3n+11)}{8}
	+2\left(\frac{n(3n+16)}{8}\right)\\
	&=2n^2+6n.
\end{align*}
\end{proof}

\begin{theorem}\label{C(M(Kn,m))}
	Let $G$ be a complete bipartite graph with $n+m$ vertices, $K_{n,m}$ where $n,m>1$. The closeness value of the middle complete bipartite graph is\\
	\[C(M(K_{n,m}))=\frac{(m+n)(m+n-1)+mn(4+6(m+n)+2mn)}{8}. \]
\end{theorem}
\begin{proof}
	In order to prove the theorem, \ we can split the vertex set of middle $%
	K_{n,m}$ graph. Let us denote the set of vertices as $\{u_{i}:i\in
	\{1,2,...,n+m\}\}\cup \{v_{j}:j\in \{1,2,...,nm\}\}$ where the set $%
	\{v_{j}:j\in \{1,2,...,nm\}\}$ represent the vertices that were later added
	and associated with edges to form the middle graph structure and $%
	\{u_{i}:i\in \{1,2,...,n+m\}\}$ represent the vertices of $K_{n,m}$ graph
	with partite sets of sizes $n$ and $m.$ Based on this, we can examine the
	state of vertices under three cases:
	
	\begin{itemize}
		\item There are $m$ vertices from the set $\{u_{i}:i\in \{1,2,...,n+m\}\}$
		at distance one from $n$ vertices, and $(n-1)$ vertices at distance three.
		The remaining $nm$ vertices are at distance two.
		
		\item There are $n$ vertices from the set $\{u_{i}:i\in \{1,2,...,n+m\}\}$
		at distance one from $m$ vertices, and $(m-1)$ vertices at distance three.
		The remaining $nm$ vertices are at distance two.
		
		\item For $nm$ vertices from the set $\{v_{j}:j\in \{1,2,...,nm\}\}$, there
		are $(m+n)$ vertices at distance one, and $(mn-1)$vertices at distance two.
		There are no vertices at distance three. Therefore, the total closeness
		value is calculated as;
		\[C(M(K_{n,m}))=\frac{(m+n)(m+n-1)+mn(4+6(m+n)+2mn)}{8}.\]         
	\end{itemize}
\end{proof}

\section{The Vertex Residual Closeness of Some Special Middle Graphs}

In this section, calculations related to the more sensitive parameter,
residual closeness, which is derived from the closeness parameter, will be
conducted.

\begin{theorem}
	Let $G$ be a path graph with $n>3$ vertices, $P_{n}$. The residual closeness
	value of the middle path graph is
\[
R(M(P_{n}))= 
\begin{cases}
	2C(M(P_{\frac{n}{2}})) & ,\text{if }n\text{ is even} \\ 
	C(M(P_{\frac{n-1}{2}}))+C(M(P_{\frac{n+1}{2}})) & ,\text{if }n\text{ is odd}
\end{cases}%
\]
\end{theorem}
\begin{proof}
	Let us notate the set of vertices of the middle path graph as $%
	\{u_{1},u_{2},\ldots ,u_{n},v_{1},v_{2},\ldots ,v_{n-1}\}$. Here, $u_{i}$,
	for $1\leq i\leq n$, represents the vertices of the $P_{n}$ graph, and $%
	v_{j} $, for $1\leq j\leq n-1$, are vertices added later according to the
	definition of the middle graph. Considering the structure of the path middle
	graph, we can obtain the residual closeness value in four different ways.
	The vertices to be removed from the graph are examined in four types as
	follows:
	\begin{itemize}
\item The first type of removal is to remove either the $u_{1}$ or $u_{n}$
		vertex,
		
\item The second type of removal is to remove one of the vertices from $%
		\{u_{2},\ldots ,u_{n-1}\}$,
		
\item The third type of removal is to remove either the $v_{1}$ or $%
		v_{(n-1)} $ vertex, and
		
\item The fourth type of removal is to remove one of the vertices from $%
		\{v_{2},\ldots ,v_{(n-2)}\}$ in the graph.
	\end{itemize}
	
Among all possible vertex removals, we show that  the fourth type of removal yields minimum, as the graph becomes disconnected and certain
	vertex-to-vertex connections are lost. To achieve the minimum value, if $n$
	is even, it is logical to remove the vertex $v_{n/2}$ from the set $%
	\{v_{2},\ldots ,v_{n-2}\}$. However, if $n$ is odd, it is reasonable to
	remove a vertex from $\{v_{(n-1)/2},$ $v_{(n+1)/2}\}$. In doing so, more
	vertex relationships will be severed. Once the closeness is recalculated
	based on the modified distances due to vertex removal, the residual
	closeness value for the path middle graph is determined as follows:
	
	Case1: When $n$ is even and the vertex $v_{n/2}$ is removed from the graph,
	the graph will split into two $M(P_{n/2})$ graphs. Based on the result from Theorem \ref{C1}, the residual closeness value can
	be computed as $2C(M(P_{n/2}))$.
	
	Case 2: When $n$ is odd, and without loss of generality, the vertex $%
	v_{(n-1)/2}$ is removed, the graph will split into two disconnected
	components: $M(P_{(n-1)/2})$ and $M(P_{(n+1)/2})$. Applying the result from
	Theorem \ref{C1}, the residual closeness value is
	obtained from $C(M(P_{\frac{n-1}{2}}))+C(M(P_{\frac{n+1}{2}}))$.
\end{proof}

\begin{theorem}
	Let $G$ be a cycle graph with $n\geq3$ vertices, $C_{n}$. The vertex residual
	closeness value of the middle cycle graph is
\[
R(M(C_{n}))=7n-16+\frac{18}{2^{n}}. 
\]
\end{theorem}
\begin{proof}
	Let us notate the set of vertices of the middle cycle graph as $%
	\{u_{1},u_{2},\ldots ,u_{n},v_{1},v_{2},\ldots ,v_{n}\}$. Here, $u_{i}$, for 
	$1\leq i\leq n$, represents the vertices of the $C_{n}$ graph, and $v_{j}$,
	for $1\leq j\leq n$, are vertices added later based on the definition of the
	middle graph. Considering the structure of the middle cycle graph, there
	will be two different types of removal as follows:
	
	\begin{itemize}
		\item the first type is to remove one vertex from $\{u_{1},u_{2},\ldots
		,u_{n}\}$, and
		
		\item the second type is to remove one vertex from $\{v_{1},v_{2},\ldots
		,v_{n}\}$ in the graph.
	\end{itemize}
	
	The second type of removal clearly leads to a minimized value. When a vertex
	is removed from the set $\{v_{1},v_{2},\ldots ,v_{n}\}$, the remaining
	structure transforms into a path graph with $n$ vertices. Using the result
	derived from Theorem \ref{C1}, we can determine the
	residual closeness value for the cycle graph. Consequently, by removing one
	vertex from $\{v_{1},v_{2},\ldots ,v_{n}\}$, residual closeness value of the
	middle cycle graph is calculated, as the resulting structure resembles a
	path graph with $n$ vertices, following the conclusion of Theorem 
	\ref{C1}. The structure becomes isomorphic to a path
	graph structure with $n$ vertices. In this case, using the result obtained from Theorem \ref{C1}, we will obtain the residual
	closeness value for the cycle graph as $R(M(C_{n}))=C(M(P_{n}))=7n-16+\frac{%
		18}{2^{n}}$.
\end{proof}

\begin{theorem}
	Let $G$ be a star graph with $n+1$ vertices, $S_{1,n}$, $n\geq2$. The residual
	closeness value of the middle star graph is
\[
R(M(S_{1,n}))=\frac{9n^{2}-7n-2}{8}. 
\]
\end{theorem}
\begin{proof}
	Let the vertex set of the middle star graph denoted by $\{u_{1},u_{2},\ldots
	,u_{(n+1)},v_{1},v_{2},\ldots ,v_{n}\}$. Here, $u_{i}$, for $1\leq i\leq n+1$%
	, represents the vertices of the $S_{1,n}$ graph where $u_{1}$ is the
	central vertex and $v_{j}$, for $1\leq j\leq n$, are vertices added later
	according to the definition of the middle graph. Let us consider the
	structure of the middle star graph, there will be three types of removal as
	follows:
	
	\begin{itemize}
		\item The first type of removal is the central vertex $u_{1}$,
		
		\item The second type of removal is choosing one vertex from $\{u_{2},\ldots
		,u_{n+1}\}$, and
		
		\item The third type of removal is choosing one vertex from $%
		\{v_{1},v_{2},\ldots ,v_{n}\}$ in the graph.
	\end{itemize}
	
	Approaching intuitively, the third type of removal will minimize the
	closeness value after removal. Therefore, the removal of one vertex from the
	set $\{v_{1},v_{2},\ldots ,v_{n}\}$ will result in disconnected component in
	the graph. However, considering the distances of the vertices $%
	\{v_{1},v_{2},\ldots ,v_{n}\}$ to the other vertices, there will be a
	reduction in the closeness value $C(M(S_{1,n}))$ obtained from Theorem 
	\ref{M(S)} by $(\frac{9n+1}{4})$.
Therefore, the residual closeness value for the middle star graph is obtained as:
\[
R(M(S_{1,n}))=\frac{9n^{2}-7n-2}{8}. 
\]
\end{proof}
\begin{theorem}
	Let $G$ be a wheel graph with $(n+1)$ vertices, $W_{1,n}, n\geq6$. The residual
	closeness value of the middle wheel graph is%
	\[
	R(M(W_{1,n}))=\frac{16n^{2}+28n+3}{8}. 
	\]
\end{theorem}

\begin{proof}
	Let the vertex set of the middle wheel graph denoted by $\{u_{1},u_{2},%
	\ldots ,u_{(n+1)},v_{1},v_{2},\ldots ,v_{2n}\}$. Here, $u_{i}$, for $1\leq
	i\leq n+1$, represents the vertices of the $W_{1,n}$ graph where $u_{1}$ is
	the central vertex and $v_{j}$, for $1\leq j\leq 2n$, are vertices added
	later according to the definition of the middle graph. For ease of
	expression, let us divide the set of vertices in the set of $\{v_{j}:1\leq
	j\leq 2n\}$ into two parts: the first $n$ vertices are adjacent to the
	central vertex, and the last $n$ vertices are those obtained from cycle form
	of the $W_{1,n}$ graph. Let us consider the structure of the middle wheel
	graph, there will be three types of removal as follows:
	
	\begin{itemize}
		\item The first type of removal is a vertex from $\{u_{i}:1\leq i\leq n+1\}$,
		
		\item The second type of removal is choosing one vertex from $%
		\{v_{1},v_{2},\ldots ,v_{n}\}$, and
		
		\item The third type of removal is choosing one vertex from $%
		\{v_{n+1},v_{n+2},\ldots ,v_{2n}\}$ in the graph.
	\end{itemize}
	
	When obtaining the residual closeness value, since the changes in the
	closeness value after removing an vertex are considered, the value obtained
	from the first bullet is simply derived by subtracting twice the closeness
	value of the removed vertex from $C(M(W_{1,n}))$. The change derived from
	the third bullet results in the distance between $v_{j+1}$ $j\in
	\{n,...2(n-1)\}$ and its neighbor $u_{i},$ increasing by $1$ after the
	removal. However, the value obtained from the second bullet is calculated as
	follows:
	
	Without loss of generality, the vertex $v_{1}$ is removed from the graph.
	Therefore, the distance between the vertex $u_{1}$ and $u_{2}$ increases
	from $2$ to $3,$ the distance between the vertex $u_{2}$ and $\{u_{5},...,$ $%
	u_{n-1}\}$ increases from $3$ to $4,$ the distance between the vertex $u_{2}$
	to $\{v_{3},...,v_{n-1}\}$ increases from $2$ to $3,$ and the distance
	between the vertex $u_{2}$ and $\{v_{n+4},..,v_{2n-3}\}$ increases from $3$
	to $4.$ Thus, the value after removal can be calculated as:%
	\[
	C(M(W_{1,n}))-2C(v_{1})+T 
	\]%
	where $C(v_{1})=\frac{4n+3}{4}$ and $T$ represents changing total closeness after removal and $T=2(-\frac{1}{2^{2}}+\frac{1%
	}{2^{3}}-\frac{n-5}{2^{3}}+\frac{n-5}{2^{4}}-\frac{n-3}{2^{2}}+\frac{n-3}{%
		2^{3}}-\frac{n-6}{2^{3}}+\frac{n-6}{2^{4}}).$ Hence, the total value is%
	\begin{equation*} 
		C(M(W_{1,n}))-\frac{20n-3}{8}.
	\end{equation*}%
	This is obvious that residual closeness value of $M(W_{1,n})$ comes from
	second bullet. When $C(M(W_{1,n}))=2n^{2}+6n$ substituted into equation
	above, the the residual closeness value is get as:%
	\[
	R(M(W_{1,n}))=\frac{16n^{2}+28n+3}{8}. 
	\]
\end{proof}

\begin{theorem}
	Let $G$ be a complete bipartite graph with $n+m$ vertices, $K_{n,m}$ where $n,m>1$. The	vertex residual closeness value of the middle complete bipartite graph is \\
	\[R(M(K_{n,m}))=\frac{(m+n)(m+n-11)+mn(6(m+n)+2mn)+5}{8}. \]
\end{theorem}

\begin{proof}
	Let the vertex set of the middle complete bipartite graph denoted by $\{u_{1},u_{2},\ldots,u_{(n+m)},v_{1},v_{2},\ldots ,v_{nm}\}$. Here, $u_{i}$%
	, for $1\leq i\leq n+m$, represents the vertices of the $K_{n,m}$ graph and $%
	v_{j}$, for $1\leq j\leq nm$, are vertices added later according to the
	definition of the middle graph. Let us consider the structure of the middle
	complete bipartite graph, there will be two types of removal as follows:
	
	\begin{itemize}
		\item The first type of removal is the any vertex from the set of $%
		\{u_{i}:1\leq i\leq n+m\}$,
		
		\item The second type of removal is any vertex from $\{v_{1},v_{2},\ldots
		,v_{nm}\}.$
	\end{itemize}
	
	The removal of the first type does not cause any change in the distances
	between the vertices. There is only a decrease in the value of $C(K_{n,m})$
	by twice the closeness value of the vertex $u_{i}$ that is removed from the
	graph. Let us assume that $n>m$ in this case, $C(i)$ $=\frac{4n+2nm+m-1}{8}$%
	. In the case of the second type of removal, there will be some reductions
	in the value of $C(K_{n,m})$ due to the changing distances between the
	vertices, aside from the decrease by $2C(v_{j})$. This change occurs as
	follows: without loss of generality, let the removed edge represent the
	vertex $v_{j}$ corresponding to the edge $(u_{1}u_{n+m})$. In this case, the
	distance between vertices $u_{1}$ and $u_{n+m}$ increases from $2$ to $4$.
	Additionally, the distance between $u_{1}$ and each $v_{j}$ vertex (where
	each edge represented in the middle graph \ as $v_{j}$ is adjacent to $%
	u_{n+m}$ in $K_{n,m})$ increases from $2$ to $3$. Similarly, the distance
	between $u_{n+m}$ and each vertex of $v_{j}$ (where each edge represented in
	the middle graph \ as $v_{j}$ is adjacent to $u_{1}$ in $K_{n,m})$ also
	increases from $2$ to $3$. Accordingly, the total change becomes as $%
	T=-(\frac{2(n-1)+2(m-1)+3}{8})$. In addition, for each vertex $v_{j},$
	the closeness value is $C(v_{j})=\frac{2(m+n)+mn-1}{4}.$ Therefore, the
	residual closeness value of the middle complete bipartite can be obtained as%
	\[
	R(M(K_{n,m}))=C(M(K_{n,m}))-2C(v_{j})+T. 
	\]%
	Using the result of Theorem \ref{C(M(Kn,m))}, we get\\
	\[R(M(K_{n,m}))=\frac{(m+n)(m+n-11)+mn(6(m+n)+2mn)+5}{8}.\] 	
\end{proof}

\section{Some General Results}

In this section, we will provide general results regarding the closeness, vertex residual closeness value of the middle graph. The results obtained will establish a relationship between the closeness value of the middle graph of $G$ and the closeness value of $G$. Additionally, some results found in the previous section will be necessary for the derived bounds.

\begin{theorem}
	Let $G$ be a connected graph of order $n\geq 2$. Then
\[
\frac{C(G)}{2}+6n-14+\frac{1}{2^{n-4}}
\leq C(M(G))
\leq
\frac{C(G)}{2}
+\frac{n(n-1)(n^{2}+7n-2)}{16},
\]
where $C(G)$ is the closeness value of $G$.
\end{theorem}
\begin{proof}
	According to the structure of the middle graph, let the vertex set be
	expressed as $V(G)\cup V(E(G))$. To determine the closeness value, we will
	consider the sum of three components: the closeness value of the vertices
	within $V(G)$, denoted by $C(V(G)\sim V(G)),$ twice the closeness value
	between the vertices in $V(G)$ and the vertices $V(E(G))$ added to form the	middle graph, denoted by $C(V(G)\sim V(E(G)))$, and the total closeness value among the vertices in $V(E(G))$, denoted by $C(V(E(G))\sim V(E(G)))$ themselves: 
	\begin{align*}
	C(M(G))&=C(V(G)\sim V(G))+2C(V(G)\sim V(E(G)))\\
    &+C(V(E(G))\sim V(E(G))). 
	\end{align*}
	Due to the definition of the middle graph, $C(V(G)\sim V(G))$ $=\frac{C(G)}{2%
	}$. Since the path minimizes closeness among connected graphs of order
	$n$, $C(V(G)\sim V(E(G)))\geq $ $C(P_{n})$
	and, hence $C(V(E(G))\sim V(E(G)))\geq $ $C(P_{n-1})$. We derive the following
	lower bound as: 
	\[
	\frac{C(G)}{2}+2C(P_{n})+C(P_{n-1})\leq C(M(G)). 
	\]%
	Using the value $C(P_{n})=2n-4+\frac{1}{2^{n-2}}$ from Theorem \ref{zeynep} (iii)%
	\[
	\frac{C(G)}{2}+6n-14+\frac{1}{2^{n-4}}\leq C(M(G)). 
	\]%
	We now establish the upper bound. Adding edges to a connected graph
	cannot increase the distances between the corresponding vertices
	already present in the middle graph. Moreover, every added edge
	introduces an additional edge-vertex and therefore additional
	nonnegative terms in
	$C(V(G)\sim V(E(G)))$ and
	$C(V(E(G))\sim V(E(G)))$.
	Consequently, these two quantities attain their maximum values when
	$G\cong K_n$.
	
	We first calculate
	$C(V(K_n)\sim V(E(K_n)))$.
	Each vertex of $K_n$ is incident with $n-1$ edges. Therefore, there
	are
	$
	n(n-1)
	$
	vertex--edge pairs at distance $1$ in $M(K_n)$. The remaining
	vertex--edge pairs are at distance $2$. Their number is
		$nm-n(n-1)
		=\frac{n(n-1)(n-2)}{2}.$ where $m=|E(K_n)|$. Hence,
	\begin{align*}
		C(V(K_n)\sim V(E(K_n)))
		&=
		\frac{n(n-1)(n+2)}{8}.
	\end{align*}
	Therefore,
	\[
	C(V(G)\sim V(E(G)))
	\leq
	\frac{n(n-1)(n+2)}{8}.
	\]
	
	Next, consider
	$C(V(E(K_n))\sim V(E(K_n)))$.
	Every edge of $K_n$ is adjacent to exactly
	$2(n-2)$
	other edges. Since $V(E(K_n))$ contains $m$ vertices, the number of
	ordered pairs of adjacent edge-vertices is
	$2m(n-2).$
	Each such pair is at distance $1$ in $M(K_n)$. The total number of ordered pairs of distinct edge-vertices is
	$m(m-1).$ Thus, the number of ordered pairs of nonadjacent edge-vertices is
	$
	m(m-1)-2m(n-2).
	$
	
	Since any two nonadjacent edges of $K_n$ have a common adjacent edge,
	the corresponding edge-vertices are at distance $2$ in $M(K_n)$.
	Therefore,
	\begin{align*}
		C(V(E(K_n))\sim V(E(K_n)))
		&=
		\frac{n(n-1)(n-2)(n+5)}{16}.
	\end{align*}
	Hence,
	\[
	C(V(E(G))\sim V(E(G)))
	\leq
	\frac{n(n-1)(n-2)(n+5)}{16}.
	\]
	Consequently,
	\begin{align*}
		C(M(G))
		&\leq
		\frac{C(G)}{2}
		+
		2\left(
		\frac{n(n-1)(n+2)}{8}
		\right)\\
		&\quad+
		\frac{n(n-1)(n-2)(n+5)}{16}\\
		&=
		\frac{C(G)}{2}
		+
		\frac{n(n-1)(n^2+7n-2)}{16}.
	\end{align*}
	
	The equality in the upper bound is attained for $G\cong K_n$. Therefore,
	\[
	\frac{C(G)}{2}
	+6n-14
	+\frac{1}{2^{n-4}}
	\leq
	C(M(G))
	\leq
	\frac{C(G)}{2}
	+\frac{n(n-1)(n^2+7n-2)}{16}.
	\]
\end{proof}
\begin{theorem}
	\label{regular}Let $G$ be $r$-regular connected graph with $n\geq3$ vertices and $nr$ is even.
	Then,%
	\[
	C(M(C_{n}))\leq C(M(G))\leq \frac{C(G)}{2}+nr(\frac{4n+4r+nr+2}{16})
	\]%
	where $C(M(C_{n}))$ is obtained Theorem \ref{M(Cn)}.
\end{theorem}

\begin{proof}
	Let $G$ be $r$-regular graph and vertex set of $M(G)$ is denoted by $%
	V(G)\cup V(E(G))$ where $V(G)=\{u_{i}:i\in \{1,...,n\}\}$ and $%
	V(E(G))=\{v_{j}:j\in \{1,...,\frac{nr}{2}\}\}.$To calculate the closeness
	value, we split the relationships between vertices as follows:
	
	\begin{itemize}
		\item The relationship among vertices $u_{i}$, $1\leq i\leq n$, total
		closeness value represented as $C(u\sim u)$.
		
		\item The relationship among vertices $v_{j}$, $1\leq j\leq \frac{nr}{2}$,
		total closeness value represented as $C(v\sim v)$.
		
		\item The relationship between vertices $u_{i}$ and $v_{j}$ and total
		closeness value represented as $C(u\sim v)$.
	\end{itemize}
	
	In this case, for the first bullet, the total closeness value is obtained as 
	$\frac{C(G)}{2}$ and, for the second bullet, each vertex $v_{j}$ has $2(r-1)$
	neighbour in the set of $\{v_{j^{\prime }}:$ $1\leq j^{\prime }\leq nr/2,$ $%
	j\neq j^{\prime }\}$ and the remaining vertices are at least at a distance
	of $2$. Based on this, $C(v\sim v)\leq \frac{nr(4r+nr-6)}{2^{4}}$, and for the third bullet, each vertex $u_{i}$ has $r$ neighbour in the
	set of $\{v_{j}:$ $1\leq j\leq nr/2\}$ and the remaining vertices are at
	least at a distance of $2$. Based on this, $C(u\sim v)\leq \frac{%
		nr(n+2)}{2^{3}}.$ From this, if we derive an upper bound for the closeness
	value of $C(M(G))$, is:%
	\begin{eqnarray*}
		C(M(G)) &=&C(u\sim u)+C(v\sim v)+2C(u\sim v) \\
		&\leq &\frac{C(G)}{2}+nr(\frac{4n+4r+nr+2}{16}).
	\end{eqnarray*}%
	The lower bound, on the other hand, is obtained from the weakest regular
	structure of relationships between vertices and edges, which is the cycle
	graph $C_{n}.$ Thus, we have%
	\[
	C(M(C_{n}))\leq C(M(G))\leq \frac{C(G)}{2}+nr(\frac{4n+4r+nr+2}{16}) 
	\]%
	where $C(M(C_{n}))$ is obtained Theorem \ref{M(Cn)}.
\end{proof}

\begin{remark}
	\label{remKn}In the previous theorem, when $r=n-1$, the graph $K_{n}$ becomes
	a complete graph. In this case, the equality in the upper bound is achieved,
	and $C(M(K_{n}))$ represents the maximum value that can be obtained for a
	regular graph. However, for other regular structures, as the value of $r$
	changes, the upper bound will narrow. Therefore, providing a result that
	depends on a general $r$ value is more practical in terms of the upper bound.
\end{remark}

\begin{theorem}\label{tree}
	Let $G$ be a tree graph with order $n\geq2$ %
	\[
	C(M(G))=\frac{5}{2}C(G)+C(L(G)) 
	\]%
	where $C(L(G))$ is the closeness value of line graph of $G.$
\end{theorem}

\begin{proof}
	Let $G$ be a tree structure and vertex set of $M(G)$ is denoted by $V(G)\cup
	V(E(G))$ where $V(G)=\{u_{i}:i\in \{1,...,n\}\}$ and $V(E(G))=\{v_{j}:j\in
	\{1,...,n-1\}\}.$ In order to obtain closeness value of $C(M(G))$, the
	relationship between vertices can be split as:
	
	\begin{itemize}
		\item The relationship among vertices $u_{i}$, $1\leq i\leq n$, total
		closeness value represented as $C(u\sim u)$.
		
		\item The relationship among vertices $v_{j}$, $1\leq j\leq n-1$, total
		closeness value represented as $C(v\sim v)$.
		
		\item The relationship between vertices $u_{i}$ and $v_{j}$ and total
		closeness value represented as $C(u\sim v)$.
	\end{itemize}
	
	Here, the value $C(u\sim u)=\frac{C(G)}{2}$ and $C(v\sim %
	v)=C(L(G))$ where $L(G)$ is the line graph of $G,$ and from the tree
	structure and definition of the middle graph $C(u\sim v)=C(G).$%
	Therefore, it is obtained that $C(M(G))=\frac{5}{2}C(G)+C(L(G)).$
\end{proof}

%

\begin{remark}
	When comparing the closeness value of the star graph provided in Theorem \ref{zeynep} with the closeness value of the star line graph provided in Theorem \ref{line}, it is evident that $C(S_{1,n})\leq
	C(L(S_{1,n})) $ for number of vertices greater than $5$. In the star structure,
	since the edges are at a distance of $1$ from each other, the closeness
	value of the line graph reaches its maximum value. In addition, By applying the results of Theorem \ref{zeynep} (ii) and Theorem \ref{line} (iii) to the general context, Theorem \ref{M(S)}—focused on the specific case of a star graph, a special type of tree—both supports and illustrates the broader theoretical framework. This coherence across general and specific cases strengthens the validity and applicability of the Theorem \ref{tree}.
\end{remark}

\begin{theorem}
	Let $G$ be a graph with $E(G)\neq\emptyset$. Then
	\[
	R(M(G))\leq \min_{e\in E(G)} C(M(G-e)).
	\]
\end{theorem}

\begin{proof}
	For each edge $e\in E(G)$, let $v_e$ denote the vertex of $M(G)$
	corresponding to $e$. By the definition of vertex residual closeness,
	\[
	R(M(G))
	=
	\min_{x\in V(M(G))} C(M(G)-x).
	\]
	Since $v_e\in V(M(G))$, it follows that
	\[
	R(M(G))
	\leq
	C(M(G)-v_e)
	\]
	for every $e\in E(G)$.
	
	Removing the vertex $v_e$ from $M(G)$ eliminates exactly the vertex
	corresponding to the edge $e$ of $G$. The incidences between the
	remaining original vertices and edge-vertices, as well as the
	adjacencies among the remaining edge-vertices, are unchanged.
	Therefore,
	\[
	M(G)-v_e\cong M(G-e).
	\]
	Since closeness is invariant under graph isomorphism,
	\[
	C(M(G)-v_e)=C(M(G-e)).
	\]
	Hence,
	\[
	R(M(G))
	\leq
	C(M(G-e))
	\]
	for every $e\in E(G)$. Taking the minimum over all edges of $G$, we
	obtain
	\[
	R(M(G))
	\leq
	\min_{e\in E(G)} C(M(G-e)).
	\]
\end{proof}

\section{Algorithm Analysis}

In this section, we present the computational procedure used to
calculate closeness and vertex residual closeness for the graph
structures considered in this paper. The aim of the algorithm is to
provide a computational verification of the theoretical results
obtained in the previous sections rather than to introduce a new
shortest-path algorithm.

The input of the procedure is the adjacency matrix $M$ of the graph
under consideration. In particular, when the procedure is applied to
a middle graph, $M$ represents the adjacency matrix of $M(G)$. The
algorithm contains four main computational components: Breadth-First
Search (BFS), the Floyd--Warshall algorithm, closeness computation,
and vertex residual closeness computation.

The BFS procedure starts from a selected vertex and determines its
shortest-path distances to the remaining vertices. Thus, it provides
a direct way of examining the distance structure and connectivity of
the input graph. In addition, the Floyd--Warshall algorithm is used
to calculate all-pairs shortest-path distances. These distances are
then used in the computation of the closeness values.

For a graph $H$, the closeness value considered throughout this paper
is
\[
C(H)=
\sum_{v_j\in V(H)}
\sum_{\substack{v_i\neq v_j\\d(v_i,v_j)<\infty}}
\frac{1}{2^{d(v_i,v_j)}}.
\]
If two vertices belong to different connected components, their
distance is treated as infinite and hence their contribution to
closeness is zero.

To obtain vertex residual closeness, each vertex is removed
successively and the closeness value of the remaining graph is
recalculated. The minimum value obtained over all possible vertex
removals gives
\[
R(H)=\min_{v_k\in V(H)} C(H-v_k).
\]

The procedure was implemented in JavaScript and was used to check
the theoretical values obtained for the graph classes considered in
this paper. The complete pseudocode is given below.

\subsection{Computational Verification}

To verify the consistency of the computational procedure with the
analytical results, the implementation was applied to representative
graphs for which closed formulas were obtained in the previous
sections. We considered $P_4$, $C_4$, and $K_4$ as illustrative
examples. The computational values were compared with the
corresponding theoretical values.

\begin{table}[ht]
	\centering
	\caption{Computational results.}
	\label{tab:computationalverification}
	\begin{tabular}{c|c|c|c|c|c}
		\hline
		$G$
		& $|V(G)|$
		& $|V(M(G))|$
		& Computed $C(M(G))$
		& Computed $R(M(G))$\\
		\hline
		$P_4$ & $4$ & $7$     & $13.125$ & $5$\\
		$C_4$ & $4$ & $8$     & $19.5$   & $13.125$\\
		$K_4$ & $4$ & $10$    & $34.5$   & $26.75$\\
		\hline
	\end{tabular}
\end{table}

For the path $P_4$, Theorem~2.1 gives
\[
C(M(P_4))
=
7(4)-16+\frac{18}{2^4}
=
13.125.
\]
The computational procedure produces exactly the same value.
Moreover, Theorem~3.1 gives
\[
R(M(P_4))
=
2C(M(P_2))
=
5,
\]
which agrees with the computational result.

For the cycle $C_4$, Theorem~2.2 gives
\[
C(M(C_4))=19.5.
\]
The value produced by the computational procedure is also $19.5$.
Furthermore, by Theorem~3.2,
\[
R(M(C_4))
=
C(M(P_4))
=
13.125,
\]
which coincides with the computed value.

Similarly, for the complete graph $K_4$, Theorem~2.4 gives
\begin{align*}
	C(M(K_4))
	&=
	\frac{4(4-1)\left(4^2+7(4)+2\right)}{16}\\
	&=
	34.5.
\end{align*}
The computational procedure gives the same closeness value. The
computed vertex residual closeness of $M(K_4)$ is $26.75$.

These examples show that the computational procedure is consistent
with the closed formulas derived theoretically. The implementation
therefore provides an independent computational check of the
analytical calculations and may also be applied to graphs for which
closed expressions are not available.

\subsection{Complexity Analysis}

We now discuss the computational complexity of the procedure. In the
pseudocode, $n$ denotes the number of vertices in the input adjacency
matrix $M$. Thus, when the algorithm is applied to a middle graph,
\[
n=|V(M(G))|.
\]
If the original graph $G$ has $p$ vertices and $q$ edges, then
\[
|V(M(G))|=p+q.
\]

When BFS is executed from a single vertex using an adjacency matrix,
all possible neighbor relationships may need to be examined.
Therefore, the time complexity of one BFS execution is
$
O(n^2).
$
If BFS is executed from every vertex in order to obtain the complete
BFS distance matrix, the total time complexity becomes
$
O(n^3).
$

The Floyd--Warshall algorithm determines the shortest-path distances
between all pairs of vertices by considering every possible
intermediate vertex. Therefore, its time complexity is
$
O(n^3).
$

After the all-pairs shortest-path matrix is obtained, the
\textsc{ClosenessCentrality} function examines all ordered pairs of
vertices. This requires
$
O(n^2)
$
additional operations. Consequently, since the Floyd--Warshall
computation dominates this step, the total time required to compute closeness is
$
O(n^3).
$

The most computationally demanding part is the
\textsc{ResidualCloseness} function. Each vertex is removed once and,
after every removal, closeness is recalculated. Since the closeness
calculation contains an $O(n^3)$ Floyd--Warshall computation and this
procedure is repeated for at most $n$ vertices, the worst-case time
complexity is
$
O(n\cdot n^3)=O(n^4).
$
Therefore, the complete procedure has worst-case time complexity
$
O(n^4),
$
where $n$ is the number of vertices of the graph represented by the
input adjacency matrix. In particular, when the procedure is applied
to a middle graph $M(G)$ of a graph $G$ with $p$ vertices and $q$
edges, this can equivalently be written as
$
O\left((p+q)^4\right).
$

The computational procedure is primarily intended to support and
verify the theoretical results obtained in this paper. The numerical
examples above demonstrate agreement between the analytical formulas
and the computational outputs for representative graph classes.
For substantially larger graph instances, more efficient
shortest-path or dynamic update techniques may be considered in
future work.

\noindent
\textbf{Algorithm 1: Graph Analysis Suite (Simplified)}

\small
\begin{algorithmic}[1]
	\State \textbf{Input:} Adjacency matrix $M$
	\State \textbf{Constants:} $INF \gets$ a very large number
	\State $n \gets$ number of nodes in $M$
	\Function{BFS}{$M$, start}
	\State $dist \gets$ array of size $n$ filled with $INF$
	\State $queue \gets$ empty queue
	\State $dist[start] \gets 0$
	\State enqueue $start$ to $queue$
	\While{$queue$ is not empty}
	\State $current \gets$ dequeue from $queue$
	\For{$neighbor = 0$ to $n-1$}
	\If{$M[current][neighbor] > 0$ and $dist[neighbor] = INF$}
	\State $dist[neighbor] \gets dist[current] + 1$
	\State enqueue $neighbor$
	\EndIf
	\EndFor
	\EndWhile
	\State \Return $dist$
	\EndFunction
	
	\Function{FloydWarshall}{$M$}
	\State $dist \gets n \times n$ matrix filled with $INF$
	\For{$i = 0$ to $n-1$}
	\For{$j = 0$ to $n-1$}
	\If{$i = j$} \State $dist[i][j] \gets 0$
	\ElsIf{$M[i][j] > 0$} \State $dist[i][j] \gets M[i][j]$
	\EndIf
	\EndFor
	\EndFor
	\For{$k = 0$ to $n-1$}
	\For{$i = 0$ to $n-1$}
	\For{$j = 0$ to $n-1$}
	\If{$dist[i][k] \ne INF$ and $dist[k][j] \ne INF$}
	\State $dist[i][j] \gets \min(dist[i][j], dist[i][k] + dist[k][j])$
	\EndIf
	\EndFor
	\EndFor
	\EndFor
	\State \Return $dist$
	\EndFunction
	
	\Function{ClosenessCentrality}{$M$}
	\State $SP \gets$ \Call{FloydWarshall}{$M$}
	\State $closeness \gets$ zero array of size $n$
	\For{$i = 0$ to $n-1$}
	\State $sum \gets 0$
	\For{$j = 0$ to $n-1$}
	\If{$i \ne j$ and $SP[i][j] \ne INF$}
	\State $sum \gets sum + 2^{-SP[i][j]}$
	\EndIf
	\EndFor
	\State $closeness[i] \gets sum$
	\EndFor
	\State \Return $closeness$
	\EndFunction
	
	\Function{ResidualCloseness}{$M$}
	\State $ckSums \gets$ empty list
	\For{$k = 0$ to $n-1$}
	\State $tempM \gets$ copy of $M$
	\State set row $k$ and column $k$ of $tempM$ to $0$
	\State $c \gets$ \Call{ClosenessCentrality}{$tempM$}
	\State append $\sum c$ to $ckSums$
	\EndFor
	\State \Return $\min(ckSums)$
	\EndFunction
	
	\State $bfsMatrix \gets$ matrix of BFS distances for each node
	\State $floydMatrix \gets$ \Call{FloydWarshall}{$M$}
	\State $closeness \gets$ \Call{ClosenessCentrality}{$M$}
	\State $residual \gets$ \Call{ResidualCloseness}{$M$}
\end{algorithmic}
\section{Conclusion}
We investigated the vulnerability and robustness of networks through the lens of residual closeness. By applying these metrics to middle graphs derived from various base graphs, we achieved a deeper understanding of how structural failures impact communication efficiency in complex networks. The analytical results indicate that residual closeness can provide useful information about performance degradation in networks, particularly in cases where the network remains connected despite internal disruptions.

Furthermore, middle graphs provide a useful representation for studying systems in which both vertex and edge interactions are relevant, such as certain infrastructure and communication network models. This dual representation enables closeness and residual closeness to be examined within a structure that incorporates both vertex and edge relationships. The results obtained for these structures illustrate the potential of middle graphs for closeness-based analysis in network models involving both types of interactions.

In addition, obtaining results related to closeness and vertex residual closeness under previously unexplored important graph structures or specific graph operations constitutes another promising research direction. Systematic investigation of such structures can help derive more general and abstract characterizations of how residual closeness behaves across various graph classes.

\end{document}